%% file: main_extended.tex
\documentclass[runningheads]{llncs}

\usepackage{amssymb}

\usepackage{amsmath}
\usepackage{stmaryrd}
\usepackage[ruled,vlined]{algorithm2e}
\usepackage{url}
\usepackage{paralist}

\newcommand{\inV}{^-}
\newcommand{\gen}{{\textsf{A}}}

\newcommand{\ew}{\varepsilon}
\newcommand{\sA}{\textsf{A}}
\newcommand{\al}{\sA}
\newcommand{\sB}{\textsf{B}}
\newcommand{\Br}{\sB}
\newcommand{\Z}{\mathbb{Z}}
\renewcommand{\L}{\mathcal{L}}
\newcommand{\T}{\mathcal{T}}
\newcommand{\F}{\mathcal{F}}
\newcommand{\D}{\mathcal{D}}
\newcommand{\sem}[1]{\llbracket#1\rrbracket}
\newcommand{\ignore}[1]{}
\newcommand{\angl}[1]{\langle#1\rangle}
\newcommand{\fgeq}{=_{\F_\gen}}
\newcommand{\dunion}{\sqcup}
\newcommand{\dconc}{\star}
\newcommand{\vars}{\mathfrak{X}}
\newcommand{\abs}[1]{\lvert#1\rvert}

\newcommand{\LTW}{${\sf LT}_\gen$}
\newcommand{\SLP}{{\sf SLP}}
\newcommand{\DTA}{{\sf DTA}}
\newcommand{\dom}{\textsf{dom}}
\newcommand{\depth}{\textsf{depth}}

\spnewtheorem*{lem}{Lemma}{\bfseries}{\itshape}
\spnewtheorem*{cor}{Corollary}{\bfseries}{\itshape}

\begin{document}

\title{Equivalence of Linear Tree Transducers with Output in the Free Group}
\author{Raphaela L\"obel\and
Michael Luttenberger
\and
Helmut Seidl}
\authorrunning{L\"obel et al.}
\institute{TU M\"unchen, Germany 
\email{\{loebel, luttenbe, seidl\}@in.tum.de}}
\maketitle              
\begin{abstract}
We show that equivalence of deterministic linear tree transducers 
can be decided in polynomial time when their outputs are interpreted over the free group.
Due to the cancellation properties offered by the free group, the required constructions
are not only more general, but also simpler than the corresponding constructions for proving
equivalence of deterministic linear tree-to-word transducers.

\keywords{Linear tree transducer \and free group \and equivalence problem \and polynomial time}
\end{abstract}

\input{intro}

\input{prelims}

\input{equiv}

\input{conclusion}

\bibliographystyle{splncs04} 
\bibliography{lit}

\input{appendix}

\end{document}

%% file: intro.tex
\section{Introduction}\label{sec:intro} 

In 2009, Staworko and Niehren observed that
equivalence for \emph{sequential} tree-to-word transducers 
\cite{Staworko2009} 
can be reduced to the morphism equivalence problem for context-free languages.
Since the latter problem is decidable in polynomial time \cite{Plandowski}, 
they thus proved that equivalence of sequential tree-to-word transducers is decidable in
polynomial time.
This decision procedure was later accompanied by a canonical normal form
which
can be applied to learning \cite{Laurence2011,Laurence2014}. 
Sequentiality of transducers means that subtrees must always be processed from left to right.
This restriction was lifted by Boiret who provided a canonical normal form for
unrestricted linear tree-to-word transducers \cite{Boiret2016}. 
Construction of that normal form, however, may require exponential time implying that
the corresponding decision procedure requires exponential time as well.
In order to improve on that, Palenta and Boiret provided a polynomial time procedure which
just normalizes the \emph{order} in which an unrestricted linear tree-to-word transducer 
processes the subtrees of its input \cite{Palenta2016}.  
They proved that after that normalization, equivalent transducers are necessarily \emph{same-ordered}.
As a consequence, equivalence of linear tree-to-word transducers can thus also be reduced 
to the morphism equivalence problem for context-free languages and thus can be decided in
polynomial time. 
Independently of that, Seidl, Maneth and Kemper showed by algebraic means, that 
equivalence of general (possibly non-linear) tree-to-word transducers is decidable \cite{SeidlMK18}.
Their techniques are also applicable if the outputs of transducers are not just in a
free monoid of words, but also if outputs are in a free group.
The latter means that output words are considered as equivalent not just when they are literally 
equal, but also when they become equal after cancellation of matching positive and 
negative occurrences of letters.
For the special case of \emph{linear} tree transducers with outputs either in a
free monoid or a free group, Seidl et al.\ provided a \emph{randomized} polynomial time
procedure for in-equivalence.
The question remained open whether for outputs in a free group, 
randomization can be omitted.
Here, we answer this question to the affirmative.
In fact, we follow the approach of \cite{Palenta2016} to normalize the order in which
tree transducers produce their outputs. For that normalization, we heavily rely on
commutation laws as provided for the free group. 
Due to these laws, the construction as well as the arguments for its correctness,
are not only more general but also much cleaner than in the case of outputs in a free monoid only.
The observation that reasoning over the free group may simplify arguments has
also been made, e.g., by Tomita and Seino and later by Senizergues 
when dealing with the equivalence problem for deterministic pushdown transducers
\cite{TomitaS1989,Senizergues1999}.
As morphism equivalence on context-free 
languages is decidable in polynomial time --- even if the morphism outputs are in a free group
\cite{Plandowski}, we obtain a polynomial time algorithm for equivalence of
tree transducers with output in the free group.

%% file: prelims.tex
 \section{Preliminaries}\label{sec:prelims}\label{s:prelims}
 
 We use $\Sigma$ to denote a finite \emph{ranked} alphabet, while $\gen$ is used for an \emph{unranked} alphabet.
 $\T_\Sigma$ denotes the set of all trees (or terms) over $\Sigma$.
 The \emph{depth} $\depth(t)$ of a tree $t\in\T_\Sigma$ equals 0, if $t= f()$ for some $f\in\Sigma$ of rank 0,
 and otherwise, $\depth(t) = 1+\max\{\depth(t_i)\mid i=1,\ldots,m\}$ for $t=f(t_1,\ldots,t_m)$.
 We denote by $\F_{\gen}$ the representation of the free group generated by $\gen$ where the carrier is the set of \emph{reduced} words instead of the usual quotient construction:
 For each $a\in\gen$, we introduce its \emph{inverse}
 $a\inV$.  The set of elements of $\F_\gen$ then consists of all words over 
 the alphabet $\{a,a\inV\mid a\in\gen\}$ which do not contain $a\,a\inV$ or $a\inV a$ as factors. 
 These words are also called \emph{reduced}.
 In particular, $\gen^*\subseteq\F_\gen$.
 The group operation ``$\cdot$'' of $\F_\gen$ is concatenation, followed by \emph{reduction}, i.e.,
 repeated cancellation of subwords $a\, a\inV$ or $a\inV a$.
 Thus, $a\,b\,c\inV\cdot c\,b\inV\,a \fgeq a\,a$.
 The \emph{neutral element} w.r.t.\ this operation is the empty word $\ew$, while the inverse $w\inV$ of
 some element $w\in\F_\gen$ is obtained by reverting the order of the letters in $w$ while 
 replacing each letter $a$ with $a\inV$ and each $a\inV$ with $a$.
 Thus, e.g., $(a\,b\,c\inV)\inV = c\,b\inV a\inV$.
 
 In light of the inverse operation $(\,.\,)\inV$, 
 we have that $v\cdot w \fgeq v'w'$ where
 $v = v'u$ (as words) for a maximal suffix $u$ so that $u\inV$ is a prefix of $w$ with $w= u\inV w'$.
 For an element $w\in\F_\gen$, $\angl{w} = \{w^l\mid l\in\Z\}$ denotes the cyclic subgroup
 of $\F_\gen$ generated from $w$. As usual, we use the convention that $w^0 = \ew$,
 and $w^{-l} = (w\inV)^l$ for $l>0$.
 An element $p\in\F_\gen$ different from $\ew$, is called \emph{primitive} if $w^l\fgeq p$ for some
 $w\in\F_\gen$ and $l\in\Z$ implies that $w\fgeq p$ or $w\fgeq p\inV$, i.e., $p$ and $p\inV$ 
 are the only (trivial) roots of $p$. 
 Thus, primitive elements generate \emph{maximal} cyclic subgroups of $\F_\gen$.
 We state two crucial technical lemmas. 
 
\begin{lemma}\label{l:todo}
Assume that $y^m \fgeq \beta\cdot y^n\cdot\beta\inV$ with $y\in\F_\gen$ primitive.
Then $m=n$, and $\beta \fgeq y^k$ for some $k\in\Z$.
\end{lemma}
\begin{proof}
Since $\beta\cdot y^n\cdot\beta\inV \fgeq (\beta\cdot y\cdot\beta\inV)^n$,
we find by \cite[Proposition 2.17]{Lyndon} a primitive element $c$ such that $y$ and 
$\beta\cdot y\cdot\beta^-$
are powers of $c$. As $y$ is primitive, $c$ can be chosen as $y$. Accordingly,
\begin{equation}
{\small y^j\fgeq\beta\cdot y\cdot \beta^-}\label{eq:conjugate}
\end{equation}
holds for some $j$.
If $\beta$ is a power of $y$, then $\beta\cdot y\cdot\beta\inV\fgeq y$, and the assertion
of the lemma holds. 
Likewise if $j=1$, then $\beta$ and $y$ commute. Since $y$ is primitive, then $\beta$ necessarily 
must be a power of $y$.

For a contradiction, therefore now assume that $\beta$ is not a power of $y$ and $j\neq 1$.
W.l.o.g., we can assume that $j>1$.
First, assume now that $y$ is \emph{cyclically reduced}, i.e., 
the first and last letters, $a$ and $b$, respectively, of $y$
are not mutually inverse. 
Then for each $n>0$, $y^n$ is obtained from $y$ by $n$-concatenation of $y$ as a word 
(no reduction taking place).
Likewise, either the last letter of $\beta$ is different $a\inV$ or the first letter of $\beta\inV$ 
is different from $b\inV$ because these two letters are mutually inverse.
Assume that the former is the case.
Then $\beta\cdot y$ is obtained by concatenation of $\beta$ and $y$ as words (no reduction taking
place).
By \eqref{eq:conjugate}, $\beta\cdot y^n\fgeq y^{j\cdot n}\cdot\beta$.
for every $n\geq 1$.
Let $m>0$ denote the length of $\beta$ as a word.
Since $\beta$ can cancel only a suffix of $y^{j\cdot n}$ of length at most $m$,
it follows, that the word $\beta\,y$ must a prefix of the word $y^{m+1}$.
Since $\beta$ is not a power of $y$, the word $y$ can be factored into $y=y'c$ for
some non-empty suffix $c$ such that $\beta=y^{j'}y'$, implying that $yc = cy$ holds. 
As a consequence, $y=c^{l}$ for some $l>1$ --- in contradiction to the irreducibility of $y$.

If on the other hand, the first letter of $\beta\inV$ is not the inverse of the last letter of $y$,
then $y\cdot\beta\inV$ is obtained as the concatenation of $y$ and $\beta\inV$ as words.
As a consequence, $y\beta\inV$ is a suffix of $y^{m+1}$, and we arrive at a contradiction.

We conclude that the statement of the lemma holds whenever $y$ is cyclically reduced.
Now assume that $y$ is not yet cyclically reduced. Then we can find a maximal suffix $r$ of $y$
(considered as a word) such that $y = r\inV sr$ holds and $s$ is cyclically reduced.
Then $s$ is also necessarily primitive. (If $s \fgeq c^n$, then $y\fgeq (r\inV cr)^n$).
Then assertion \eqref{eq:conjugate} can be equivalently formulated as
$$
{\small s^j\fgeq(r\cdot\beta\cdot r\inV)\cdot y\cdot (r\cdot\beta\cdot r\inV)^-}
$$
We conclude that $r\cdot\beta\cdot r\inV\fgeq s^l$ for some $l\in\Z$. But then
$\beta \fgeq (r\inV\cdot s\cdot r)^l \fgeq y^l$, and the claim of the lemma follows.
\end{proof}

\begin{lemma}\label{l:periodicGroup}
Assume that $x_1,x_2$ and $y_1,y_2$ are distinct elements in $\F_\gen$ and that
\begin{equation}
{\small x_i\cdot\alpha\cdot y_j\cdot\beta \fgeq \gamma\cdot y'_j\cdot\alpha'\cdot x'_i\cdot\beta'}
\label{eq:commute}
\end{equation}
holds for $i=1,2$ and $j=1,2$.
Then there is some primitive element $p$ and exponents $r,s \in \Z$ such that
$x_1\cdot\alpha \fgeq x_2\cdot\alpha\cdot p^r$ and $y_1 \fgeq p^s\cdot y_2$.
\end{lemma}
\begin{proof}
By the assumption \eqref{eq:commute},
\[
{\small\begin{array}{lll}
\gamma	&\fgeq&(x_1\cdot\alpha\cdot y_j\cdot\beta)\cdot
               (y'_j\cdot\alpha'\cdot x'_1\cdot\beta')\inV 	\\
	&\fgeq&(x_2\cdot\alpha\cdot y_j\cdot\beta)\cdot
               (y'_j\cdot\alpha'\cdot x'_2\cdot\beta')\inV 	\\
\end{array}}
\]
for all $j=1,2$. Thus,
\[
{\small\begin{array}{lll@{\quad}l}
x_1\cdot\alpha\cdot y_j\cdot\beta{\beta'}^- {x'}^-_1{\alpha'}^-{y'}^-_j	&\fgeq&
x_2\cdot\alpha\cdot y_j\cdot\beta\cdot{\beta'}^-\cdot {x'}^-_2\cdot{\alpha'}^-\cdot{y'}^-_j
&\text{\normalsize implying}	\\
y_j^-\cdot\alpha^-\cdot x_2^-\cdot x_1\cdot\alpha\cdot y_j	&\fgeq& 
\beta\cdot{\beta'}^-\cdot {x'}^-_2\cdot x'_1\cdot\beta'\cdot\beta^-
\end{array}}
\]
for $j=1,2$. Hence,
\[
{\small\begin{array}{lll@{\quad}l}
y_1\inV\cdot\alpha\inV\cdot x_2\inV\cdot x_1\cdot \alpha\cdot y_1	&\fgeq& 
y_2\inV\cdot\alpha\inV\cdot x_2\inV\cdot x_1\cdot \alpha\cdot y_2
&\text{\normalsize implying}\\
(x_2\cdot\alpha)\inV x_1\cdot\alpha &\fgeq& 	
(y_1\cdot y_2\inV)\cdot((x_2\cdot\alpha)\inV\cdot x_1\cdot\alpha)\cdot(y_1\cdot y_2\inV)\inV 
\end{array}}
\]
Since $x_1$ is different from $x_2$, also $x_1\cdot\alpha$ is different from $x_2\cdot\alpha$.
Let $p$ denote a primitive root of $(x_2\cdot\alpha)\inV\cdot x_1\cdot\alpha$.
Then by Lemma \ref{l:todo},
\[
{\small\begin{array}{lll}
x_1\cdot\alpha &\fgeq& x_2\cdot\alpha\cdot p^r	\\
y_1 &\fgeq& p^s\cdot y_2
\end{array}}
\]
for suitable exponents $r,s\in\Z$.
\end{proof}
As the elements of $\F_\gen$ are words, they can be represented by 
\emph{straight-line} programs (\SLP s). An \SLP\ is a context-free grammar
where each non-terminal occurs as the left-hand side of exactly one rule.
We briefly recall basic complexity results for operations on elements of $\F_\gen$ 
when represented as \SLP s \cite{DBLP:books/daglib/0038981}.
\begin{lemma}\label{l:SLPops}
Let $U, V$ be \SLP s representing words $w_1, w_2 \in \{a, a^- \mid a \in \gen\}$, respectively.
Then the following computations/decision problems can be realized in polynomial time
\begin{compactitem}
\item compute an \SLP\ for $w_1^{-}$;
\item compute the primitive root of $w_1$ if $w_1\neq\ew$;
\item compute an \SLP\ for $w \fgeq w_1$ with $w$ reduced;
\item decide whether $w_1 \fgeq w_2$;
\item decide whether it exists $g \in \F_\gen$, such that $w_1 \in g\cdot\angl{w_2}$ and
compute an SLP for such $g$.
\end{compactitem}
\end{lemma}
In the following, we introduce \emph{deterministic linear tree transducers} which produce
outputs in the free group $\F_\gen$. For convenience, we follow the approach in 
\cite{SeidlMK18} where only \emph{total} deterministic transducers are considered ---
but equivalence is relative w.r.t.\ some \emph{top-down deterministic} domain automaton $B$.
A top-down deterministic automaton (\DTA) $B$ is a tuple $(H,\Sigma,\delta_B,h_0)$ where
$H$ is a finite set of states,
$\Sigma$ is a finite ranked alphabet,
$\delta_B:H\times\Sigma\to H^*$ is a partial function where $\delta_B(h,f)\in H^k$ if the
the rank of $f$ equals $k$, and
$h_0$ is the start state of $B$.
For every $h\in H$, we define the set $\dom(h)\subseteq\T_\Sigma$ by
$f(t_1,\ldots,t_m)\in\dom(h)$ iff $\delta_B(h,f) = h_1\ldots h_m$ and
$t_i\in\dom(h_i)$ for all $i=1,\ldots,k$. 
$B$ is called \emph{reduced} if $\dom(h)\neq\emptyset$ for all $h\in H$.
The \emph{language} $\L(B)$ accepted by $B$ is the set $\dom(h_0)$.
We remark that for every \DTA\ $B$ with $\L(B)\neq\emptyset$, 
a reduced \DTA\ $B'$ can be constructed in polynomial time with $\L(B) = \L(B')$.
Therefore, we subsequently assume w.l.o.g.\ that each \DTA\ $B$ is reduced.

A (total deterministic) \emph{linear tree transducer} with output in $\F_\gen$ (\LTW\ for short) 
is a tuple $M = (\Sigma, \gen, Q, S, R)$ where 
 $\Sigma$ is the ranked alphabet for the input trees, 
 $\gen$ is the finite (unranked) output alphabet,
 $Q$ is the set of \emph{states}, 
 $S$ is the \emph{axiom} of the form $u_0$ or $u_0\cdot q_0(x_0)\cdot u_1$ with $u_0, u_1 \in \F_\gen$ and
 $q_0\in Q$,
 and $R$ is the set of \emph{rules} which contains for each state $q\in Q$ and each
 input symbol $f\in\Sigma$, one rule of the form
\begin{equation}
{\small q(f(x_1, \ldots, x_m)) \to u_0\cdot q_1(x_{\sigma(1)})\cdot\ldots\cdot u_{n-1}\cdot
                                       q_n(x_{\sigma(n)})\cdot u_n}
	\label{eq:rule}
\end{equation}
 Here, $m$ is the rank of $f$, $n \leq m$, $u_0,\ldots,u_n\in\F_\gen$ and 
 $\sigma$ is an injective mapping from $\{1,\ldots, n\}$ to $\{1, \ldots, m\}$.
 The \emph{semantics} of a state $q$ is the function $\sem{q}:\T_\Sigma\to\F_\gen$ 
defined by
 $$\sem{q}(f(t_1,\ldots,t_m)) \fgeq u_0 \cdot\sem{q_1}(t_{\sigma(1)}) \cdot\ldots\cdot u_{n-1}\cdot \sem{q_n}(t_{\sigma(n)})\cdot u_n$$
 if there is a rule of the form \eqref{eq:rule}
 in $R$.
 Then the \emph{translation} of $M$ is the function $\sem{M}:\T_\Sigma\to\F_\gen$ defined by
 $\sem{M}(t) \fgeq u_0$ if the axiom of $M$ equals $u_0$, and
 $\sem{M}(t) \fgeq u_0\cdot \sem{q}(t)\cdot u_1$ if the axiom of $M$ is given by 
$u_0\cdot q(x_0)\cdot u_1$.

 \begin{example}\label{ex:LT}
 Let $\gen = \{a,b\}$.  As a running example we consider the \LTW\ $M$ with 
 input alphabet $\Sigma=\{f^2, g^1, k^0\}$ where the superscripts indicate the rank of the input symbols.
 $M$ has axiom $q_0(x_0)$ and the following rules
\[\arraycolsep10pt
{\small \begin{array}{lll}
q_0(f(x_1, x_2)) \to q_1(x_2) b q_2(x_1) & 
q_0(g(x_1)) \to q_0(x_1)  	& 
q_0(k) \to \ew 	\\
q_1(f(x_1, x_2)) \to q_0(x_1) q_0(x_2) & 
q_1(g(x_1)) \to ab q_1(x_1) 	& 
q_1(k) \to a 	\\
q_2(f(x_1, x_2)) \to q_0(x_1) q_0(x_2) &
q_2(g(x_1)) \to ab q_2(x_1) &
 q_2(k) \to ab
\end{array}}
\]
 \end{example}
\noindent
 Two \LTW s $M$, $M'$ are \emph{equivalent} relative to the \DTA\ $B$ iff their translations coincide
 on all input trees accepted by $B$, i.e., $\sem{M}(t)\fgeq\sem{M'}(t)$ for all $t\in\L(B)$.

To relate the computations of the \LTW\ $M$ and the domain automaton $B$,
we introduce the following notion.
A mapping $\iota:Q\to H$ from the set of states of $M$ to the set of states of $B$
is called \emph{compatible} if either the set of states of $M$ is empty (and thus the axiom of 
$M$ consists of an element of $\F_\gen$ only), or the following holds:
\begin{enumerate}
\item	$\iota(q_0) = h_0$;
\item	If $\iota(q) = h$, $\delta_B(h,f) = h_1\ldots h_m$, and
	there is a rule in $M$ of the form \eqref{eq:rule}
	then $\iota(q_i) = h_{\sigma(i)}$ for all $i=1,\ldots,n$;
\item	If $\iota(q) = h$ and $\delta_B(h,f)$ is undefined for some $f\in\Sigma$ of rank $m\geq 0$,
	then $M$ has the rule $q(f(x_1,\ldots,x_m))\to\bot$ for some dedicated symbol $\bot$
	which does not belong to $\gen$.
\end{enumerate}
\begin{lemma}\label{l:rho}
For an \LTW\ $M$ and a \DTA\ $B=(H,\Sigma,\delta_B,h_0)$, an \LTW\ $M'$  with 
a set of states $Q'$ together with a mapping $\iota:Q'\to H$ can be constructed
in polynomial time such that the following holds:
\begin{enumerate}
\item	$M$ and $M'$ are equivalent relative to $B$;
\item	$\iota$ is compatible.
\end{enumerate}
\end{lemma}

\begin{example}\label{ex:compatible}
Let \LTW\ $M$ be defined as in Example~\ref{ex:LT}.
Consider \DTA\ $B$ with start state $h_0$ and the transition function
$\delta_B=\{(h_0,f) \mapsto h_1 h_1, (h_1,g) \mapsto h_1, (h_1,h) \to \ew\}$.
According to Lemma \ref{l:rho}, \LTW\ $M'$ for $M$ then is defined as follows.
$M'$ has axiom {\small$\angl{q_0, h_0}(x_0)$} and the rules
\[
{\small \begin{array}{l@{\qquad}l}
\multicolumn{2}{l}{\angl{q_0, h_0}(f(x_1, x_2)) \to \angl{q_1, h_1}(x_2)\, b\, \angl{q_2, h_1}(x_1)} \\
\angl{q_1, h_1}(g(x_1)) \to ab\, \angl{q_1, h_1}(x_1) 	& \angl{q_1, h_1}(k) \to a 	\\
\angl{q_2, h_1}(g(x_1)) \to ab\, \angl{q_2, h_1}(x_1) & \angl{q_2, h_1}(k) \to ab	\\
\end{array} }
\]
where the rules with left-hand sides {\small$\angl{q_0, h_0}(g(x_1))$, $\angl{q_0, h_0}(h)$, $\angl{q_1, h_1}(f(x_1, x_2))$, $\angl{q_2, h_1}(f(x_1, x_2))$},
all have right-hand-sides $\bot$.
The compatible map $\iota$ is then given by 
{\small$\iota = \{\angl{q_0,h_0}\mapsto h_0, \angl{q_1,h_1}\mapsto h_1,\angl{q_2,h_1}\mapsto h_1\}$}.
For convenience, we again denote the pairs $\angl{q_0,h_0},\angl{q_1,h_1},\angl{q_2,h_1}$ with
$q_0,q_1,q_2$, respectively.
\end{example}
Subsequently, we w.l.o.g.\ assume that each \LTW\ $M$ with corresponding
\DTA\ $B$ for its domain, comes with a compatible map $\iota$.
Moreover, we define for each state $q$ of $M$, the set $\L(q) = \{\sem{q}(t)\mid t\in\dom(\iota(q))\}$
of all outputs produced by state $q$ (on inputs in $\dom(\iota(q))$),
and $\L^{(i)}(q)=\{\sem{q}(t)\mid t\in\dom(\iota(q)),\depth(t)<i\}$ for $i\geq 0$.

Beyond the availability of a compatible map, we also require that all states of $M$ are 
\emph{non-trivial} (relative to $B$).
Here, a state $q$ of $M$ is called \emph{trivial} if $\L(q)$ contains a single element only.
Otherwise, it is called \emph{non-trivial}.
This property will be established in Theorem \ref{t:order}.

%% file: equiv.tex
\section{Deciding Equivalence}\label{sec:equiv}\label{s:equiv}

In the first step, we show that equivalence relative to the \DTA\ $B$ of \emph{same-ordered} \LTW s is decidable.
\ignore{
, i.e., \LTW s that process their input in the same way,
can be decided in polynomial time by a reduction to the morphism equivalence problem of 
context-free grammars over the free group. 
A similar construction was used in \cite{Staworko2009} to decide equivalence of
\emph{sequential} tree-to-word transducers considering word equivalence.
Then, we define an \emph{order} for \LTW s that guarantees that two equivalent \LTW s are same-ordered if both are ordered.
To obtain our main result we therefore show that from every \LTW\ an equivalent \LTW\ can be constructed that is ordered.
This procedure is based on \cite{Palenta2016} where the equivalence of linear tree-to-word transducers
concerning word equivalence was shown
to be decidable in polynomial time via a partial normal form.

First, we present a \emph{direct} reduction for \emph{same-ordered} \LTW s 
to the morphism equivalence problem over the free group.
}
For a \DTA\ $B$,
consider the \LTW s $M$ and $M'$ with compatible mappings $\iota$ and $\iota'$, respectively.
$M$ and $M'$ are \emph{same-ordered} relative to $B$
if they process their input trees in the same order.
We define set of pairs $\angl{q,q'}$ of \emph{co-reachable} states of $M$ and $M'$.
Let $u_0\cdot q_0(x_1)\cdot u_1$ and
$u'_0\cdot q'_0(x_1)\cdot u'_1$ be the axioms of $M$ and $M'$, respectively,
where $\iota(q_0) = \iota'(q'_0)$ is the start state of $B$.
Then the pair $\angl{q_0,q'_0}$ is co-reachable.
Let $\angl{q,q'}$ be a pair of co-reachable states. Then $\iota(q)=\iota'(q')$ should
hold. 
For $f \in \Sigma$,
assume that $\delta_B(\iota(q),f)$ is defined.
Let
\begin{equation}
{\small\begin{array}{lll}
q(f(x_1, \ldots, x_m)) &\to& u_0 q_1(x_{\sigma(1)}) u_1 \ldots u_{n-1} q_n(x_{\sigma(n)}) u_n\\
q'(f(x_1, \ldots, x_m)) &\to& u'_0 q'_1(x_{\sigma'(1)}) u'_1 \ldots u'_{n-1} q'_n(x_{\sigma'(n')}) u'_{n'}
\end{array}} \label{eq:coreach}
\end{equation}
be the rules of $q, q'$ for $f$, respectively.
Then $\angl{q_{j},q'_{j'}}$ is co-reachable whenever $\sigma(j)=\sigma'(j')$ holds.
In particular, we then  have $\iota(q_j) = \iota'(q'_{j'})$.

The pair $\angl{q,q'}$ of co-reachable states is called same-ordered, if
for each corresponding pair of rules \eqref{eq:coreach},
$n = n'$ and $\sigma = \sigma'$.
Finally, $M$ and $M'$ are \emph{same-ordered} if for every co-reachable pair $\angl{q,q'}$ 
of states of $M, M'$, and every $f\in\Sigma$,
each pair of rules \eqref{eq:coreach} is same-ordered
whenever $\delta_B(\iota(q),f)$ is defined.

Given that the \LTW s $M$ and $M'$ are same-ordered relative to $B$, 
we can represent the set of pairs of runs of 
$M$ and $M'$ on input trees by means of a single context-free grammar $G$.
The set of nonterminals of $G$ consists of a distinct start nonterminal $S$ together with 
all co-reachable pairs $\angl{q,q'}$ of states $q,q'$ of $M,M'$, respectively.
The set of terminal symbols $T$ of $G$ is given by $\{a,a\inV,\bar a,\bar a\inV\mid a\in\gen\}$
for fresh distinct symbols $\bar a,\bar a\inV,a\in\gen$.
Let $\angl{q,q'}$ be a co-reachable pair of states of $M, M'$,
and $f \in\Sigma$ such that $\delta_B(\iota(q),f)$ is defined.
For each corresponding pair of rules \eqref{eq:coreach},
$G$ receives the rule
$$\angl{q,q'} \to u_0 \bar u'_0 \angl{q_1, q'_1} u_1 \bar u'_1 \ldots u_{n-1} \bar u'_{n-1} \angl{q_n,q'_n} u_n \bar u'_n$$
where $\bar u'_i$ is obtained from $u'_i$ by replacing each output symbol
$a\in\gen$ with its barred copy $\bar a$ as well as each inverse $a\inV$ with its barred copy $\bar a\inV$.
For the axioms $u_0 q(x_1) u_1$ and $u'_0 q'(x_1) u'_1$
of $M,M'$, respectively,
we introduce the rule
$S \to u_0 \bar u'_0 \angl{q,q'} u_1 \bar u'_1$ where again $\bar u'_i$ are the
barred copies of $u'_i$.
We define morphisms $f,g:T^*\to\F_\gen$ by
\[\arraycolsep5pt
{\small\begin{array}{lll}
f(a) = a & f(a\inV) = a\inV & f(\bar a)=f(\bar a\inV)=\ew \\
g(\bar a) = a & g(\bar a\inV) = a\inV & g(a)=g(a\inV)=\ew \\
\end{array}}
\]
for $a\in\gen$.
Then $M$ and $M'$ are equivalent relative to $B$ iff
$g(w) \fgeq f(w)$ for all $w \in \L(G)$. 
Combining Plandowski's polynomial construction of a test set for a context-free language to check 
morphism equivalence over finitely generated free groups \cite[Theorem 6]{Plandowski}, with
Lohrey's polynomial algorithm for checking equivalence of \SLP s over the free group 
\cite{Lohrey2004}, we deduce that
the equivalence of the morphisms $f$ and $g$ on all words generated by the context-free grammar $G$,
is decidable in polynomial time.
Consequently, we obtain:

\begin{corollary}\label{cor:equiSameOrdered}
Equivalence of same-ordered \LTW s relative to a \DTA\ $B$ is decidable in polynomial time. 
\qed
\end{corollary}

\noindent
Next, we observe that for every \LTW\ $M$ with compatible map $\iota$ and
non-trivial states only, a \emph{canonical} ordering can be established.
We call $M$ \emph{ordered} (relative to $B$) if
for all rules of the form \eqref{eq:rule},
with $\L(q_i)\cdot u_i\cdot \ldots\cdot u_{j-1}\cdot \L(q_j) \subseteq v\cdot \angl{p}$, 
$p \in \F_\gen$ the ordering $\sigma(i) < \ldots < \sigma(j)$ holds.
Here we have naturally extended the operation ``$\cdot$'' to sets of elements.

We show that two ordered \LTW s, when they are equivalent, are necessarily \emph{same-ordered}.
The proof of this claim is split in two parts. 
First, we prove that the \emph{set} of indices of subtrees processed by equivalent co-reachable
states are identical and second, that the order is the same.

\begin{lemma}\label{l:sameCalls}\label{l:orderedSame}
Let $M, M'$ be \LTW s with compatible maps
$\iota$ and $\iota'$, respectively, and non-trivial states only so that $M$ and $M'$ 
are equivalent relative to the \DTA\ $B$.
Let $\angl{q,q'}$ be a pair of co-reachable states of $M$ and $M'$.
Assume that 
$\delta_B(\iota(q),f)$ is defined for some $f\in\Sigma$ and consider the
corresponding pair of rules \eqref{eq:coreach}.
Then the following holds:
\begin{enumerate}
\item
$
\{\sigma(1),\ldots, \sigma(n)\} = \{\sigma'(1),\ldots, \sigma'(n')\}
$;
\item	$\sigma = \sigma'$.
\end{enumerate}
\end{lemma}
\begin{proof}
Since $\angl{q,q'}$ is a co-reachable pair of states, there are elements
$\alpha, \alpha', \beta, \beta' \in \F_\gen$ such that
$$
\alpha\cdot \sem{q}(t)\cdot\beta \fgeq \alpha'\cdot\sem{q'}(t)\cdot\beta'
$$
holds for all $t \in \dom(\iota(q))$.
Consider the first statement.
Assume for a contradiction that $q_k(x_{j})$ occurs
on the right-hand side of the rule for $q$ but $x_{j}$ does not occur on the 
right-hand side of the rule for $q'$.
Then, there are input trees $t = f(t_1,\ldots, t_m)$ and $t' = f(t'_1,\ldots, t'_m)$,
both in $\dom(\iota(q))$,
such that $\sem{q_k}(t_j) \not\fgeq \sem{q_k}(t'_j)$ and $t_i = t'_i$ for all $i \neq j$.
Moreover, there are $\mu_1, \mu_2\in\F_\gen$ s.t.
$$
{\small \alpha\cdot\sem{q}(t)\cdot\beta 
	\fgeq \alpha\cdot\mu_1\cdot\sem{q_k}(t_{j})\cdot\mu_2\cdot\beta 
	\not\fgeq \alpha\cdot\mu_1\cdot\sem{q_k}(t'_{j})\cdot\mu_2\cdot\beta 
	\fgeq \alpha\cdot\sem{q}(t')\cdot\beta}
$$
But then,
$$
{\small \alpha\cdot\sem{q}(t)\cdot\beta \fgeq \alpha'\cdot\sem{q'}(t)\cdot\beta'
	\fgeq \alpha'\cdot\sem{q'}(t')\cdot\beta' 
	\fgeq \alpha\cdot\sem{q} (t')\cdot\beta}
$$
--- a contradiction. By an analogous argument for some $x_j$ only occurring in the right-hand
side of the rule for $q'$ the first statement follows.

Assume for contradiction that the mappings $\sigma$ and $\sigma'$ in the corresponding rules
\eqref{eq:coreach} differ.
Let $k$ denote the minimal index so that $\sigma(k) \neq \sigma'(k)$.
W.l.o.g., we assume that $\sigma'(k) < \sigma(k)$.
By the first statement, $n=n'$ and $\{\sigma(1),\ldots, \sigma(n)\} = \{\sigma'(1),\ldots, \sigma'(n)\}$.
Then there are $\ell, \ell' > k$ such that
\[
\sigma'(k)=\sigma(\ell) < \sigma(k)=\sigma'(\ell')
\]
Let $t = f(t_1,\ldots, t_n) \in \dom(\iota(q))$ be an input tree.
For that we obtain 
\[
{\small\begin{array}[t]{l}
\mu_0 := u_0\cdot\sem{q_1}(t_{\sigma(1)})\cdot\ldots\cdot u_{k-1} \\
\mu_1 := u_k\cdot\sem{q_{k+1}}(t_{\sigma(k+1)})\cdot\ldots\cdot u_{\ell-1} \\
\mu_2 := u_\ell\cdot\sem{q_\ell}(t_{\sigma(\ell)})\cdot\ldots\cdot u_n 
\end{array}\quad
\begin{array}[t]{l}
\mu'_0 := u'_0\cdot\sem{q'_1}(t_{\sigma'(1)})\cdot\ldots\cdot u'_{k-1} \\
\mu'_1 := u'_k\cdot\sem{q'_{k+1}}(t_{\sigma'(k+1)})\cdot\ldots\cdot u'_{\ell'-1} \\
\mu'_2 := u'_{\ell'}\cdot\sem{q'_{\ell'}}(t_{\sigma'(\ell')})\cdot\ldots\cdot u'_n
\end{array}}
\]
Then for all input trees $t' \in \dom(\iota(q_k))$, $t'' \in \L(\dom(\iota(q'_k))$,
{\small$$
\alpha\cdot\mu_0\cdot\sem{q_{k}}(t')\cdot
           \mu_1\cdot\sem{q_{\ell}}(t'')\cdot
           \mu_2\cdot\beta
\fgeq \alpha'\cdot\mu'_0\cdot\sem{q'_{k}}(t'')\cdot
		  \mu'_1\cdot\sem{q'_{\ell'}}(t')\cdot
		  \mu'_2\cdot\beta'
$$}
Let $\gamma' = \mu_0^-\alpha\inV\alpha'\mu'_0$. Then
{\small$$
\sem{q_k}(t')\cdot\mu_1\cdot\sem{q_\ell}(t'')\cdot\mu_2\cdot\beta 
\fgeq \gamma'\cdot\sem{q'_k}(t'')\cdot\mu'_1\cdot
		  \sem{q'_{\ell'}}(t')\cdot\mu'_2\cdot\beta'
$$}
By Lemma~\ref{l:periodicGroup}, we obtain that for all
$w_1, w_2 \in \L(q_k)$ and $v_1, v_2 \in \L(q_\ell)$,
$w_2\inV\cdot w_1 \in \mu_1\cdot\angl{p}\cdot\mu_1\inV$ and
$v_1\cdot v_2\inV \in \angl{p}$ for some primitive $p$.

If $\ell = k+1$, i.e., there is no further state
between $q_k(x_{\sigma(k)})$ and $q_\ell(x_{\sigma(\ell)})$, then 
$\mu_1 \fgeq u_k$,
$\L(q_k) \subseteq w\cdot u_k\cdot \angl{p}\cdot u_k\inV$ and
$\L(q_\ell) \subseteq \angl{p}\cdot v$
for some fixed $w \in \L(q_k)$ and $v \in \L(q_\ell)$.
As $\sigma(k) > \sigma'(k) = \sigma(\ell)$, this contradicts $M$ being ordered.

\noindent For the case that there is at least one occurrence of a state between
$q_k(x_{\sigma(k)})$ and $q_\ell(x_{\sigma(\ell)})$,
we show that for all 
$\alpha_1, \alpha_2 \in u_k\cdot\L(q_{k+1})\cdot\ldots\cdot u_{\ell-1} \fgeq: \hat{L}$,
$\alpha_1^-\alpha_2 \in \angl{p}$ holds.
We fix $w_1, w_2 \in \L(q_k)$ and $v_1,v_2\in\L(q_\ell)$ with $w_1 \neq w_2$ and $v_1\neq v_2$.
For every $\alpha \in \hat{L}$, we find by Lemma~\ref{l:periodicGroup},
primitive $p_\alpha$ and exponent $r_\alpha\in\Z$ such that 
$v_1\cdot v_2\inV \fgeq p_\alpha^{r_\alpha}$ holds.
Since $p_\alpha$ is primitive, this means that $p_\alpha\fgeq p$ or $p_\alpha\fgeq p\inV$.
Furthermore, there must be some exponent $r'_\alpha$ such that
 $w_1\inV\cdot w_2 \fgeq \alpha\cdot p^{r'_\alpha}\cdot\alpha\inV$.
For $\alpha_1,\alpha_2 \in \hat{L}$, we therefore have that
$${\small
p^{r'_{\alpha_1}} \fgeq (\alpha_1\inV\cdot\alpha_2)\cdot p^{r'_{\alpha_2}}\cdot
		(\alpha_1\inV\cdot\alpha_2)\inV
}$$
Therefore by Lemma~\ref{l:todo},
$\alpha_1\inV\cdot\alpha_2 \in \angl{p}$.
Let us fix some $w_k \in \L(q_k)$,
$\alpha \in \hat{L} \fgeq u_k\cdot\L(q_{k+1})\cdot\ldots\cdot u_{\ell-1}$,
and $w_l \in \L(q_l)$.
Then $\L(q_k) \subseteq w_k\cdot\alpha\cdot\angl{p}\cdot\alpha\inV$, 
     $\hat{L} \subseteq \alpha\cdot\angl{p}$ and
     $\L(q_l) \subseteq \angl{p}\cdot w_l$.
Therefore, 
$${\small
\L(q_k)\cdot u_k\cdot\ldots\cdot\L(q_\ell) \subseteq 
w_k\cdot\alpha\cdot\angl{p}\cdot\alpha\inV\cdot\alpha\cdot\angl{p}\cdot\angl{p}\cdot w_l 
\fgeq w_k\cdot\alpha\cdot\angl{p}\cdot w_l
}$$
As $\sigma(k) > \sigma'(k) = \sigma(\ell)$, this again contradicts $M$ being ordered.
\end{proof} 
\noindent
It remains to show that every \LTW\ can be ordered in polynomial time.
For that, we rely on the following characterization.

\begin{lemma}\label{l:ultimatelyPeriodic}
Assume that $L_1,\ldots,L_n$ are neither empty nor singleton subsets of $\F_\gen$
and $u_1, \ldots, u_{n-1} \in \F_\gen$.
Then there are $v_1,\ldots,v_n\in\F_\gen$ such that
\begin{equation}
{\small
L_1\cdot u_1\cdot \ldots\cdot 
L_{n-1}\cdot u_{n-1}\cdot 
L_n \subseteq v\cdot \angl{p}} 
\label{eq:period}
\end{equation}
holds if and only if for $i=1,\ldots,n$,
$L_i \subseteq v_i\cdot\angl{p_i}$ with 
$$
{\small\begin{array}{lll}
p_n &\fgeq& p\\
p_i &\fgeq& (u_i\cdot v_{i+1})\cdot p_{i+1}\cdot (u_i\cdot v_{i+1})\inV
	\quad\text{\normalsize for }i<n
\end{array}}
$$
and 
\begin{equation}
{\small v\inV \cdot v_1\cdot u_1\cdot \ldots \cdot v_{n-1}\cdot u_{n-1}\cdot v_n\in\angl{p}}
\label{eq:periodInclusion}
\end{equation}
\end{lemma}
\begin{proof}
Let $s_1 = \ew$. For $i = 2,\ldots, n$ we fix some word 
$s_i \in L_1\cdot u_1\cdot L_2\cdot \ldots\cdot L_{i-1}\cdot u_{i-1}$.
Likewise, let $t_n = \ew$ and for $i = 1,\ldots, n-1$ fix some word 
$t_i \in u_i\cdot L_{i+1}\cdot \ldots\cdot L_{n}$, and define
$v_i \fgeq s_i\inV \cdot v\cdot t_i\inV$.

First assume that the inclusion \eqref{eq:period} holds.
Let $p'_i\fgeq t_i\cdot p\cdot t_i\inV$.
Then for all $i$, $s_i\cdot L_i\cdot t_i \subseteq v\cdot \angl{p}$, and therefore
{\small$$
L_i \subseteq 
	s_i\inV\cdot v\cdot \angl{p}\cdot t_i\inV \fgeq 
	s_i\inV\cdot v\cdot t_i\inV\cdot t_i\cdot \angl{p}\cdot t_i\inV \fgeq 
	v_i \angl{p'_i}
$$}
We claim that $p'_i = p_i$ for all $i=1,\ldots,n$.
We proceed by induction on $n-i$. 
As $t_n = \ew$, we have that $p'_n = p = p_n$.
For $i<n$, we can rewrite
$t_i \fgeq u_i\cdot w_{i+1}\cdot t_{i+1}$ where $w_{i+1}\in L_{i+1}$ and thus is of the form
$v_{i+1}\cdot p_{i+1}^{k_{i+1}}$ for some exponent $k_{i+1}$.
\[
{\small\begin{array}{lll}
p'_i &\fgeq& t_i\cdot p\cdot t_i\inV \\
     &\fgeq& u_i\cdot w_{i+1}\cdot t_{i+1}\cdot p\cdot t_{i+1}\inV\cdot w_{i+1}\inV\cdot u_i\inV\\
     &\fgeq& u_i\cdot w_{i+1}\cdot p_{i+1} \cdot w_{i+1}\inV\cdot u_i\inV
\qquad\text{\normalsize by I.H.}\\
     &\fgeq& u_i\cdot v_{i+1}\cdot p_{i+1} \cdot v_{i+1}\inV\cdot u_i\inV\\
     &\fgeq& p_i
\end{array}}
\]
It remains to prove the inclusion \eqref{eq:periodInclusion}.
Since $w_i\in L_i$, we have by \eqref{eq:period} that
$v\inV w_1\cdot u_1\cdot\ldots w_n\cdot u_n\in \angl{p}$ holds.
Now we calculate:
\[
{\small\begin{array}{lll}
	v\inV\cdot w_1\cdot u_1\cdot\ldots u_{n-1}\cdot w_n 
&\fgeq& v\inV\cdot v_1\cdot p_1^{k_1}\cdot u_1\cdot\ldots\cdot u_{n-1} \cdot v_n\cdot p_n^{k_n}	\\
&\fgeq& v\inV\cdot v_1\cdot u_1\cdot v_2\cdot p_2^{k_1+k_2}\cdot u_2\cdot\ldots\cdot u_{n-1} \cdot v_n\cdot p_n^{k_n}	\\
&\ldots&		\\
&\fgeq& v\inV\cdot v_1\cdot u_1\cdot\ldots v_{n-1}\cdot u_{n-1}\cdot v_n\cdot p_n^k
\end{array}}
\]
where $k= k_1+\ldots+k_n$. Since $p_n=p$, the claim follows.

The other direction of the claim of the lemma follows directly:
\[
{\small\begin{array}{lll}
L_1 u_1 \ldots L_{n-1} u_{n-1} L_n  
&\subseteq& v_1\cdot \angl{p_1}\cdot u_1\cdot \ldots\cdot v_{n-1}\cdot \angl{p_{n-1}}\cdot u_{n-1}\cdot v_n\cdot \angl{p_n} \\
&\fgeq& v_1\cdot u_1\cdot v_2\cdot \angl{p_2}\cdot \angl{p_2}\cdot u_2\cdot \ldots\cdot v_{n-1}\cdot \angl{p_{n-1}}\cdot u_{n-1}\cdot v_n\cdot \angl{p_n} \\
&\fgeq& v_1\cdot u_1\cdot v_2\cdot \angl{p_2}\cdot u_2\cdot \ldots\cdot v_{n-1}\cdot \angl{p_{n-1}}\cdot u_{n-1}\cdot v_n\cdot \angl{p_n} \\
&\cdots&	\\
&\fgeq& v_1\cdot u_1\cdot v_2\cdot \ldots\cdot u_{n-1}\cdot v_n\cdot \angl{p_n} \\
&\fgeq& v_1\cdot u_1\cdot v_2\cdot \ldots\cdot u_{n-1}\cdot v_n\cdot \angl{p} \\
&\subseteq&v\cdot\angl{p}
\end{array}}
\]
where the last inclusion follows from \eqref{eq:periodInclusion}.
\end{proof}

\noindent
Let us call a non-empty, non-singleton
language $L\subseteq\F_\gen$ \emph{periodic}, if $L\subseteq v\cdot\angl{p}$ for some
$v,p\in\F_\gen$. 
Lemma \ref{l:ultimatelyPeriodic} then implies that if a concatenation of languages 
and elements from $\F_\gen$ is periodic, then so must be all non-singleton component
languages. In fact, the languages in the composition can then be arbitrarily permuted.

\begin{corollary}\label{c:swap}
Assume for non-empty, nonsingleton languages $L_1,\ldots,L_n\subseteq\F_\gen$
and $u_1, \ldots, u_{n-1} \in \F_\gen$ that
property \eqref{eq:period} holds.
Then for every permutation $\pi$, there are elements $u_{\pi,0},\ldots,u_{\pi,n}\in\F_\gen$
such that 
{\small$$L_1\cdot u_1\cdot \ldots\cdot L_{n-1}\cdot u_{n-1}\cdot L_n =
u_{\pi,0}\cdot L_{\pi(1)}\cdot u_{\pi,1}\cdot \ldots\cdot u_{\pi_n-1}\cdot L_{\pi(n)}\cdot u_{\pi,n}
$$}
\end{corollary}

\begin{example}
We reconsider \LTW\ $M'$ and \DTA\ $B$ from Example~\ref{ex:compatible}.
\ignore{
To ease notation we remove the second component from the states,
thus $M'$ has states $q_0, q_1, q_2$ and rules
\[
{\small \begin{array}{lll}
q_0(f(x_1, x_2)) \to q_1(x_2) b q_2(x_1) & q_1(f(x_1, x_2)) \to \bot & q_2(f(x_1, x_2)) \to \bot \\
q_0(g(x_1)) \to \bot 	& q_1(g(x_1)) \to ab q_1(x_1) 	& q_2(g(x_1)) \to ab q_2(x_1) \\
q_0(h) \to \bot 	& q_1(h) \to a 	& q_2(h) \to ab
\end{array}}
\]
}
We observe that 
$\L(q_1) \subseteq a\cdot \angl{ba}$, $\L(q_2) \subseteq \angl{ab}$, and thus 
$\L(q_0) = \L(q_1)\cdot b\cdot \L(q_2) \subseteq \angl{ab}$.
Accordingly, the rule for state $q_0$ and input symbol $f$ is not ordered. 
Following the notation of Corollary~\ref{c:swap},
we find $v_1 = a$, $u_1 = b$ and $v_2 = \ew$, and the rule for $q_0$ and $f$ can be reordered to
$${\small q_0(f(x_1, x_2)) \to ab\cdot q_2(x_1)\cdot a\inV\cdot q_1(x_2)}$$
This example shows major improvements compared to the construction in~\cite{Palenta2016}.
Since we have inverses at hand, only \emph{local} changes must be applied
to the  sub-sequence $q_1(x_2)\cdot b\cdot q_2(x_1)$.
In contrast to the construction in~\cite{Palenta2016},
neither auxiliary states nor further changes to the rules of $q_1$ and $q_2$ are required.
\end{example}

\noindent
By Corollary~\ref{c:swap}, the order of occurrences of terms $q_k(x_{\sigma(k)})$
can be permuted in every sub-sequence
$q_i(x_{\sigma(i)})\cdot u_i\cdot \ldots \cdot u_{j-1} q_j(x_{\sigma(j)})$
where $\L(q_i)\cdot u_i\cdot \ldots\cdot u_{j-1}\cdot \L(q_j) \in u\cdot\angl{p}$ is periodic,
to satisfy the requirements of an ordered \LTW.
A sufficient condition for that is, according to Lemma~\ref{l:ultimatelyPeriodic}, 
that $\L(q_k)$ is periodic for each $q_k$ occurring in that sub-sequence.
Therefore we will determine the subset of \emph{all} states $q$ where
$\L(q)$ is periodic, and if so elements $v_q,p_q$ such that $\L(q)\subseteq v_q\cdot\angl{p_q}$.
In order to do so we compute an \emph{abstraction} of the sets $\L(q)$ by means of a
complete lattice which 
both reports constant values and also captures periodicity.

Let $\D = 2^{\F_\gen}$ denote the complete lattice of subsets of the free group $\F_\gen$.
We define a \emph{projection} $\alpha:\D\to\D$ by 
$\alpha(\emptyset)=\emptyset$, $\alpha(\{g\}) = \{g\}$, and for languages $L$ with
at least two elements,
$${\small \alpha(L) = \begin{cases}
			g\angl{p} & \text{if } L \subseteq g\angl{p} \text{ and $p$ is primitive} \\
			\F_\gen & \text{otherwise}
		\end{cases}}$$
The projection $\alpha$ is a \emph{closure} operator, i.e., is a monotonic function with
$L\subseteq\alpha(L)$, and $\alpha(\alpha(L)) = \alpha(L)$.
The image of $\alpha$ can be considered as an \emph{abstract} complete lattice $\D^\sharp$,
partially ordered by subset inclusion.
Thereby, the abstraction $\alpha$ commutes with least upper bounds as well as with
the group operation. For that, we define \emph{abstract} versions 
$\dunion,\dconc:(\D^\sharp)^2\to\D^\sharp$ 
of set union and the group operation by
\[
A_1\dunion A_2 = \alpha(A_1\cup A_2)\qquad
A_1\dconc A_2 = \alpha(A_1\cdot A_2)
\]
In fact, ``$\dunion$'' is the least upper bound operation for $\D^\sharp$.
The two abstract operators can also be more explicitly defined by:
\[
{\small\begin{array}{lllll}
\emptyset\dunion L &=& L\dunion\emptyset &=& L	\\
\F_\gen\dunion L &=& L\dunion\F_\gen &=& \F_\gen	\\
\{g_1\}\dunion\{g_2\} &=& 
\multicolumn{3}{l}{
\begin{cases}
	\{g_1\} & \text{if } g_1 = g_2	\\
	g_1\cdot\angl{p} & \text{if } g_1\neq g_2, p \text{ primitive root of } g_1\inV\cdot g_2 
	\end{cases}}	\\
\{g_1\}\dunion g_2\cdot\angl{p} &=&
g_2\cdot\angl{p}\dunion \{g_1\} &=&\begin{cases}
	g_2\cdot\angl{p} & \text{if } g_1\in g_2\cdot\angl{p}	\\
	\F_\gen	& \text{otherwise}
	\end{cases}	\\
g_1\cdot\angl{p_1}\dunion g_2\cdot\angl{p_2} &=&
\multicolumn{3}{l}{
\begin{cases}
	g_1\cdot\angl{p_1} & \text{if } p_2\in\angl{p_1} \text{ and } 
				        g_2\inV\cdot g_1\in\angl{p_1}	\\
	\F_\gen & \text{otherwise}
	\end{cases}}	\\
\end{array}}
\]
\ignore{
\begin{lemma}\label{l:exact_lub}
For all $L_1,L_2\subseteq\F_\gen$,
\[
\alpha(L_1\cup L_2) = \alpha(L_1)\dunion\alpha(L_2)
\]
\end{lemma}
}

\[
{\small\begin{array}{lllll}
\emptyset\dconc L &=& L\dconc\emptyset &=& \emptyset	\\
\F_\gen\dconc L &=& L\dconc\F_\gen &=& F_\gen\qquad\text{for } L\neq\emptyset	\\
\{g_1\}\dconc\{g_2\}	&=& \multicolumn{3}{l}{\{g_1\cdot g_2\}}		\\
\{g_1\}\dconc g_2\cdot\angl{p}	&=& \multicolumn{3}{l}{(g_1\cdot g_2)\cdot\angl{p}}	\\
g_1\cdot\angl{p}\dconc \{g_2\}	&=& \multicolumn{3}{l}{(g_1\cdot g_2)\cdot
				\angl{g_2\inV\cdot p\cdot g_2}}	\\
g_1\cdot\angl{p_1}\dconc g_2\cdot\angl{p_2}&=&\multicolumn{3}{l}{\begin{cases}
		(g_1\cdot g_2)\cdot\angl{p_2} & 
			\text{if } g_2\inV\cdot p_1\cdot g_2\in\angl{p_2}	\\
		\F_\gen		& \text{otherwise}
		\end{cases}}
\end{array}}
\]
\ignore{
\begin{lemma}\label{l:exact_conc}
For all $L_1,L_2\subseteq\F_\gen$,
\[
\alpha(L_1\cdot L_2) = \alpha(L_1)\dconc\alpha(L_2)
\]
\end{lemma}
}
\begin{lemma}\label{l:hom}
For all subsets $L_1,L_2\subseteq\F_\gen$,
$\alpha(L_1\cup L_2)	= \alpha(L_1)\dunion\alpha(L_2)$ and
$\alpha(L_1\cdot L_2) = \alpha(L_1)\dconc\alpha(L_2)$.
\end{lemma}

\noindent 
We conclude that $\alpha$ in fact represents a \emph{precise} abstract interpretation
in the sense of \cite{DBLP:books/sp/MullerOlm06}. 
\ignore{
meaning
that $\alpha(L)$ for a context free language $L$ can be computed directly by iterating 
over the abstract complete lattice $\D^\sharp$. 
More precisely}
Accordingly, we obtain:
\begin{lemma}\label{l:periodicLang}
For every \LTW\ $M$ and \DTA\ $B$ with compatible map $\iota$, 
the sets $\alpha(\L(q))$, $q$ state of $M$,
can be computed in polynomial time.
\end{lemma}
\begin{proof}
We introduce one unknown $X_q$ for
every state $q$ of $M$, and one constraint for each rule of $M$ of the form
\eqref{eq:rule} where $\delta(\iota(q),f)$ is defined in $B$. This constraint is given by:
\begin{equation}
{\small
X_q\sqsupseteq u_0\dconc X_{q_1}\dconc \ldots\dconc u_{n-1}
                                    \dconc X_{q_n}\dconc u_n}	\label{eq:periodic_constraint}
\end{equation}
As the right-hand sides of the constraints \eqref{eq:periodic_constraint} all represent
monotonic functions, the given system of constraints has a \emph{least} solution.
In order to obtain this solution, we consider for each state $q$ of $M$, the sequence
$X_{q}^{(i)},i\geq 0$ of values in $\D^\sharp$ where $X_q^{(0)} = \emptyset$, 
and for $i>0$, we set
$X_q^{(i)}$ as the least upper bound of the values obtained from
the constraints with left-hand side $X_q$ of the form \eqref{eq:periodic_constraint}
by replacing the unknowns $X_{q_j}$ on the right-hand side with the values $X_{q_j}^{(i-1)}$.
By induction on $i\geq 0$, we verify that for all states $q$ of $M$,
\[
{\small X_q^{(i)} = \alpha(\L^{(i)}(q)})
\]
holds. Note that the induction step thereby, relies on Lemma \ref{l:hom}.
As each strictly increasing chain of elements in $\D^\sharp$ consists of at most four elements,
we have that the least solution of the constraint system is attained after at most
$3\cdot N$ iterations, if $N$ is the number of states of $M$, i.e., for each state $q$ of $M$,
$X_q^{(3N)} = X_q^{(i)}$ for all $i\geq 3N$.
The elements of $\D^\sharp$ can be represented by \SLP s where the operations $\dconc$ and $\dunion$
run in polynomial time, cf.\ Lemma \ref{l:SLPops}.
Since each iteration requires only a polynomial number of operations $\dconc$ and $\dunion$,
the statement of the lemma follows.
\end{proof}

\noindent
We now exploit the information provided by the $\alpha(\L(q))$ to remove trivial states as well
as order subsequences of right-hand sides which are periodic.

\begin{theorem}\label{t:order}
Let $B$ be a \DTA\ such that $\L(B)\neq\emptyset$.
For every \LTW\ $M$ with compatible map $\iota$,
an \LTW\ $M'$ with compatible map $\iota'$ can be constructed in polynomial time such
that
\begin{enumerate}
\item	$M$ and $M'$ are equivalent relative to $B$;
\item	$M'$ has no trivial states; 
\item	$M'$ is ordered.
\end{enumerate}
\end{theorem}

\begin{proof}
By Lemma \ref{l:periodicLang}, we can, in polynomial time, determine for every state $q$ of $M$,
the value $\alpha(\L(q))$. 
We use this information to remove from $M$ all trivial states.
W.l.o.g., assume that the axiom of $M$ is given by $u_0\cdot q_0(x_0)\cdot u_1$.
If the state $q_0$ occurring in the axiom of $M$ is trivial with $\L(q_0) = \{v\}$,
then $M_1$ has no states or rules, but the axiom $u_0\cdot v\cdot u_1$.

Therefore now assume that $q_0$ is non-trivial.
We then construct an \LTW\ $M_1$ whose set of states $Q_1$
consists of all \emph{non-trivial} states $q$ of $M$ where the compatible map 
$\iota_1$ of $M_1$ is obtained from $\iota$ by restriction to $Q_1$.
Since $\L(M)\neq\emptyset$, the state of $M$ occurring in the axiom is non-trivial.
Accordingly, the axiom of $M$ is also used as axiom for $M_1$.
Consider a non-trivial state $q$ of $M$ and $f\in\Sigma$.
If $\delta(\iota(q),f)$ is not defined $M_1$ has the rule
$q(f(x_1,\ldots,x_m)\to\bot$.
Assume that $\delta(\iota(q),f)$ is defined and $M$ has a rule of the form \eqref{eq:rule}.
Then $M_1$ has the rule
$$
{\small q(f(x_1,\ldots,x_m))\to u_0\cdot g_1\cdot\ldots\cdot u_{n-1}\cdot g_n\cdot u_n}
$$
where for $i=1,\ldots,n$, $g_i$ equals $q_i(x_{\sigma(i)})$ if $q_i$ is non-trivial,
and equals the single word in $\L(q_i)$ otherwise.
Obviously, $M$ and $M_1$ are equivalent relative to $B$ where $M_1$ now has no
trivial states, while for every non-trivial state $q$, the semantics of $q$ in $M$ and $M_1$
are the same relative to $B$.
Our goal now is to equivalently rewrite the right-hand side of each rule of $M_1$
so that the result is ordered.
For each state $q$ of the \LTW\ we determine whether there are $v, p \in \Br^*$
such that $\L(q) \subseteq v\angl{p}$, cf.\ Lemma~\ref{l:periodicLang}.
So consider a rule of $M_1$ of the form \eqref{eq:rule}.
By means of the values $\alpha(\L(q_i))$, $i=1,\ldots,n$, together with the
abstract operation ``$\dconc$'', we can determine maximal intervals $[i,j]$ such that
$\L(q_i)\cdot u_i\cdot\ldots\cdot u_{j-1}\cdot\L(q_j)$ is periodic,
i.e., 
$\alpha(\L(q_i))\dconc u_i\cdot\ldots\dconc u_{j-1}\dconc\alpha(\L(q_j)) \subseteq v\cdot\angl{p}$
for some $v,p\in\F_\gen$.
We remark that these maximal intervals are necessarily disjoint.
By Corollary~\ref{c:swap}, for every permutation 
$\pi:[i,j]\to[i,j]$, elements $u',u'_i,\ldots,u'_j,u''\in\F_\gen$
can be found so that 
$q_i(x_{\sigma(i)})\cdot u_i\cdot\ldots\cdot u_{j-1}\cdot q_j(x_{\sigma(j)})$ is equivalent to
$u'\cdot q_{\pi(i)}(x_{\sigma(\pi(i))})\cdot u'_i\cdot\ldots\cdot u'_{j-1}\cdot 
q_{\pi(j)}(x_{\sigma(\pi(j))})\cdot u''$.

In particular, this is true for the permutation $\pi$ with
$\sigma(\pi(i)) < \ldots <\sigma(\pi(j))$. 
Assuming that all group elements are represented as \SLP s, 
the overall construction runs in polynomial time.
\end{proof}

\noindent
In summary, we arrive at the main theorem of this paper.

\begin{theorem}\label{t:equivGroup}
The equivalence of \LTW s relative to some \DTA\ $B$ can be decided in polynomial time.
\end{theorem}

\begin{proof}
Assume we are given \LTW s $M,M'$ with compatible maps (relative to $B$).
By Theorem \ref{t:order}, we may w.l.o.g.\ assume that $M$ and $M'$ both have no trivial states
and are ordered.
It can be checked in polynomial time whether or not $M$ and $M'$ are same-ordered.
If they are not, then by Lemma \ref{l:orderedSame}, they cannot be equivalent relative to $B$.
Therefore now assume that $M$ and $M'$ are same-ordered.
Then their equivalence relative to $B$ is decidable in polynomial time by 
Corollary \ref{cor:equiSameOrdered}.
Altogether we thus obtain a polynomial decision procedure for equivalence of \LTW s 
relative to some \DTA\ $B$.
\end{proof}

%% file: conclusion.tex
\section{Conclusion}\label{s:conclusion}

We have shown that equivalence of \LTW s relative to a given \DTA\ $B$ can be decided in polynomial
time. 
For that, we considered \emph{total} transducers only, but defined the domain of allowed input trees
separately by means of the \DTA. This does not impose any restriction of generality, since 
any (possibly partial) linear deterministic top-down tree transducer can be translated in polynomial time
to a corresponding \emph{total} \LTW\ together with a corresponding \DTA\ (see, e.g., \cite{SeidlMK18}).
The required constructions for \LTW s which we have presented here, turn out to be more general
than the constructions provided in \cite{Palenta2016} since they apply to transducers which may not
only output symbols $a\in\gen$, but also their inverses $a\inV$.
At the same time, they are \emph{simpler} and easier to be proven correct due to the 
combinatorial and algebraic properties provided by the free group.

%% file: appendix.tex
\appendix

\section{Proof of Lemma~\ref{l:rho}}
\begin{lem}
For an \LTW\ $M$ and a \DTA\ $B=(H,\Sigma,\delta_B,h_0)$, an \LTW\ $M'$  with 
a set of states $Q'$ together with a mapping $\iota:Q'\to H$ can be constructed
in polynomial time such that the following holds:
\begin{enumerate}
\item	$M$ and $M'$ are equivalent relative to $B$;
\item	$\iota$ is compatible.
\end{enumerate}
\end{lem}
\begin{proof}
In case that the axiom of $M$ is in $\F_\gen$, we obtain $M'$ from $M$ using the axiom of $M$
and using empty sets of states and rules, respectively.
Assume therefore that the axiom of $M$ is of the form $u_0\cdot q_0(x_0)\cdot u_1$.
Then \LTW\ $M'$ is constructed as follows.
The set $Q'$ of states of $M'$ consists of pairs $\angl{q,h}$, $q\in Q, h\in H$ where 
$\iota(\angl{q,h}) = h$. In particular, $\angl{q_0,h_0}\in Q'$. As the axiom of $M'$ we then 
use $u_0\cdot\angl{q_0,h_0}(x_0)\cdot u_1$.
For a state $\angl{q,h}\in Q'$, consider each input symbol $f\in\Sigma$. 
Let $m\geq 0$ denote the rank of $f$.
If $\delta_B(h,f)$ is not defined, $M'$ has the rule
$$\angl{q,h}(f(x_1,\ldots,x_m))\to\bot$$
Otherwise, let $\delta_H(h,f) = h_1\ldots h_m$, and assume that $M$ has a rule of the form 
\eqref{eq:rule}.
Then we add the states $\angl{q_i,h_{\sigma(i)}}$ to $Q'$ together with the rule
$$\angl{q,h}(f(x_1,\ldots,x_m))\to u_0\cdot \angl{q_1,h_{\sigma(1)}}(x_{\sigma(1)})\cdot \ldots
			     \cdot u_{n-1}\cdot\angl{q_n,h_{\sigma(n)}}(x_{\sigma(n)})\cdot u_n$$
By construction, the mapping $\iota$ is compatible. 
We verify for each state $\angl{q,h}$ of $M'$
and each input tree $t\in\dom(h)$
that $\sem{q}(t) = \sem{\angl{q,h}}(t)$ holds.
This proof is by induction on the structure of $t$.
From that, the equivalence of $M$ and $M'$ relative to $B$ follows.
\end{proof}

\section{Proof of Corollary~\ref{c:swap}}
\begin{cor}
Assume for non-empty, nonsingleton languages $L_1,\ldots,L_n\subseteq\F_\gen$
and $u_1, \ldots, u_{n-1} \in \F_\gen$ that
property \eqref{eq:period} holds.
Then for every permutation $\pi$, there are elements $u_{\pi,0},\ldots,u_{\pi,n}\in\F_\gen$
such that 
{\small$$L_1\cdot u_1\cdot \ldots\cdot L_{n-1}\cdot u_{n-1}\cdot L_n =
u_{\pi,0}\cdot L_{\pi(1)}\cdot u_{\pi,1}\cdot \ldots\cdot u_{\pi_n-1}\cdot L_{\pi(n)}\cdot u_{\pi,n}
$$}
\end{cor}
\begin{proof}
For $i=1,\ldots,n$, let $v_i$ and $p_i$ be defined as in Lemma \ref{l:ultimatelyPeriodic}.
Then for all $i$, $L_i \subseteq v_i \angl{p_i}$. Moreover, the languages $L'_i$
defined by $L'_n = v_n\inV\cdot L_n$ and for $i<n$,
\[
{\small L'_i = 	(u_{i}\cdot v_{i+1}\cdot\ldots\cdot u_{n-1}\cdot v_n)\inV\cdot
		(v_i\inV\cdot L_i)\cdot
		(u_{i}\cdot v_{i+1}\cdot\ldots\cdot u_{n-1}\cdot v_n)}
\]
all are subsets of $\angl{p}$. Therefore their compositions can arbitrarily be permuted.
At the same time,
\[
{\small 
L_1\cdot u_1\cdot \ldots \cdot L_{n-1}\cdot u_{n-1}\cdot L_n 
\fgeq
v_1\cdot u_1\cdot \ldots \cdot v_{n-1}\cdot u_n\cdot v_n\cdot
L'_1\cdot\ldots L'_n
}
\]
From that, the corollary follows.
\end{proof}

\section{Proof of Lemma~\ref{l:hom}}
\begin{lem}
For all subsets $L_1,L_2\subseteq\F_\gen$,
$\alpha(L_1\cup L_2)	= \alpha(L_1)\dunion\alpha(L_2)$ and
$\alpha(L_1\cdot L_2) = \alpha(L_1)\dconc\alpha(L_2)$.
\end{lem}
\begin{proof}
As $\emptyset \cup L = L \cup \emptyset = \emptyset$, it follows that
$\alpha(\emptyset \cup L) = \alpha(L \cup \emptyset) = \alpha(\emptyset) = \emptyset
					= \emptyset \dunion L' = L' \dunion \emptyset
                                                = \alpha(\emptyset) \dunion \alpha(L) = \alpha(L) \dunion \alpha(\emptyset)$.

Assume that $\alpha(L_1) = \F_\al$.
Let $L_2$ be some language, then
$\alpha(L_1 \cup L_2) = \alpha(L_2 \cup L_1) = \F_\al$ and
$\alpha(L_1) \dunion \alpha(L_2) = \F_\al \dunion \alpha(L_2) = \F_\al = \alpha(L_2) \dunion \F_\al = \alpha(L_2) \dunion \alpha(L_1)$.
The case where $\alpha(L_2)=\F_\al$ is analogous.

For $\alpha(L_1) = \{g_1\}$, $\alpha(L_2) = \{g_2\}$, both languages are singleton, and we obtain that
$\{g_1\} \cup \{g_2\} = \{g_1\}$ if and only if $g_1 = g_2$.
Accordingly, $\alpha(L_1 \cup L_2) = \alpha(\{g_1\}) = \{g_1\}=\alpha(\{g_1\}) \dunion \alpha(\{g_2\})$.
If $g_1 \neq g_2$ then $\{g_1\} \cup \{g_2\} \subseteq g_1\angl{g_1^-g_2}$
and $\alpha(L_1 \cup L_2) = g_1\angl{p}$ with $p$ the primitive root of $g_1^- g_2$.
Therefore, $\alpha(L_1 \cup L_2) = g_1\angl{p} = \alpha(\{g_1\}) \dunion \alpha(\{g_2\})$.

Assume that $\alpha(L_1)=\{g_1\}$ and $\alpha(L_2)=g_2\angl{p_2}$ for some primitive $p_2$. 
If $g_1\in g_2\angl{p_2}$, then $\alpha(L_1\cup L_2)= g_2\angl{p_2} = \alpha(L_1)\dunion\alpha(L_2)$.
Otherwise, i.e., if $g_1 \not\in g_2\angl{p_2}$, then
$L_1\cup L_2$ is not contained in $g\angl{p}$ for any $p$ (since $p_2$ was chosen primitive),
and therefore, $\alpha(L_1\cup L_2) = \F_\al = \alpha(L_1)\dunion\alpha(L_2)$.
A similar argument applies if $\alpha(L_2) = \{g_1\}$, and $\alpha(L_1) = g_2\angl{p_2}$.

Assume that $\alpha(L_1)=g_1\angl{p_1}$ and $\alpha(L_2)=g_2\angl{p_2}$ for some primitive $p_1,p_2$. 
If $p_2\in\angl{p_1}$ as well as $g_2^-g_1\in\angl{p_1}$, 
then $g_1\angl{p_1} = g_2\angl{p_2}$ (due to primitivity of $p_1,p_2$).
Moreover, $\alpha(L_1\cup L_2) = g_1\angl{p_1} = \alpha(L_1)\dunion\alpha(L_2)$.
Otherwise, i.e., if $p_2\not\in\angl{p_1}$ or $g_2^- g_1\not\in\angl{p_1}$, then
$L_1\cup L_2$ cannot be subset of $g\angl{p}$ for any $g,p\in\F_\al$.
Therefore, $\alpha(L_1\cup L_2) = \F_\al = \alpha(L_1)\dunion\alpha(L_2)$.

For the concatenation with the empty set and the product operator we find
$\alpha(\emptyset \cdot L) = \alpha(L \cdot \emptyset) = \alpha(\emptyset) = \emptyset 
= \alpha(\emptyset) \dconc \alpha(L) = \alpha(L) \dconc \alpha(\emptyset)$.

Assume that $\alpha(L_1) =\F_\al$. Then $L_1 \not\subseteq g\angl{p}$ for any $g,p\in\F_\al$.
Assume that $L_2\subseteq\F_\al$ is nonempty.
Then by Lemma~\ref{l:ultimatelyPeriodic}, $L_1 \cdot L_2$ and $L_2\cdot L_1$ cannot be contained in 
$g'\angl{p'}$ for any $g',p'$. Therefore,
$\alpha(L_1 \cdot L_2) = \alpha(L_2 \cdot L_1) = \F_\al =
\alpha(L_1) \dconc \alpha(L_2) = \alpha(L_2) \dconc \alpha(L_1)$.

For $\alpha(L_1)=\{g_1\}, \alpha(L_2)=\{g_2\}$, both languages are singletons,
and we obtain
$\alpha(L_1\cdot L_2) = \{g_1g_2\} = \{g_1\} \dconc \{g_2\} = \alpha(L_1) \dconc \alpha(L_2)$.

Now assume that $\alpha(L_1)=\{g_1\}$ and $\alpha(L_2)=g_2\angl{p_2}$.
Then $L_1=\{g_1\}$, while $L_1\cdot L_2$ is not a singleton language, but contained in
$g_1g_2\angl{p_2}$.
Therefore,
$\alpha(L_1\cdot L_2) = g_1g_2 \angl{p_2} = \alpha(L_1) \dconc \alpha(L_2)$.
Likewise, if $\alpha(L_1)=g_1\angl{p_1}$ and $\alpha(L_2)=\{g_2\}$, then
$L_2=\{g_2\}$, and $L_1\cdot L_2$ is a non-singleton language contained in
$g_1g_2\angl{g_2^-p_1g_2}$. Therefore, 
$\alpha(L_1\cdot L_2) = g_1g_2 \angl{g_2^-p_1g_2} = \alpha(L_1) \dconc \alpha(L_2)$.

Finally, let $\alpha(L_1) = g_1\angl{p_1}$ and $\alpha(L_2) = g_2\angl{p_2}$
be both ultimately periodic languages.
By Lemma~\ref{l:ultimatelyPeriodic}, $L_1 \cdot L_2$ is ultimately periodic
if and only if $g_2^-p_1g_2 \in \angl{p_2}$.
Thus if $L_1\cdot L_2$ is ultimately periodic, then
$\alpha(L_1 \cdot L_2) = g_1 g_2 \angl{p_2} = \alpha(L_1) \dconc \alpha(L_2)$.
Otherwise, $L_1 \cdot L_2 \not\subseteq g\angl{p}$ for any $g,p\in\F_\al$, and therefore
$\alpha(L_1 \cdot L_2) = \F_\al = \alpha(L_1) \dconc \alpha(L_2)$.
\end{proof}